\documentclass{llncs}
\usepackage{graphics,latexsym,amsfonts,amsmath}
\usepackage[dvips]{graphicx}
\usepackage{amsmath,amssymb}
\usepackage{dsfont}
\usepackage{stmaryrd}
\usepackage{xspace}
\usepackage{multirow}
\usepackage{color}
\usepackage{amssymb}
\usepackage{algorithm2e}
\usepackage[font=footnotesize]{subfig}

\title{On the Link Between
Strongly Connected Iteration Graphs and   
  Chaotic Boolean Discrete-Time Dynamical Systems}
\author{J. M. Bahi\inst{1}  \and J.-F. Couchot\inst{1}  \and C. Guyeux\inst{1}
  \and A. Richard\inst{2} 
}

\institute{
Computer Science Laboratory (LIFC), University of Franche-Comt\'{e}, France
\email{\{jacques.bahi,jean-francois.couchot,christophe.guyeux\}@univ-fcomte.fr}
\and
I3S - UMR CNRS 6070  Sophia Antipolis - France\\
\email{richard@i3s.unice.fr}}

\newcommand{\vectornorm}[1]{\ensuremath{\left|\left|#1\right|\right|_2}}

\newcommand{\Nats}[0]{\ensuremath{\mathbb{N}}}
\newcommand{\R}[0]{\ensuremath{\mathbb{R}}}
\newcommand{\Bool}[0]{\ensuremath{\mathds{B}}}

\newcommand{\AR}[1]{\begin{color}{red}\end{color}}
\newcommand{\JFC}[1]{\begin{color}{green}\textit{}\end{color}}
\newcommand{\CG}[1]{\begin{color}{blue}\textit{}\end{color}}

\begin{document}

\maketitle

\begin{abstract}
Chaotic  functions  are  characterized  by  sensitivity  to  initial  conditions,
transitivity, and regularity.  Providing new  functions with such properties is a
real challenge.  This work shows that one can associate with any Boolean network
a continuous function, whose discrete-time iterations are chaotic if and only if
the  iteration  graph of  the  Boolean  network  is strongly  connected.   Then,
sufficient  conditions  for  this  strong  connectivity  are  expressed  on  the
interaction graph of  this network, leading to a  constructive method of chaotic
function  computation.   The whole  approach  is  evaluated  in the  chaos-based
pseudo-random number generation context.

\end{abstract}


\section{Introduction}
Chaos  has attracted  a  lot of  attention  in various  domains  of science  and
engineering,        \textit{e.g.},        hash       function~\cite{yi2005hash},
steganography~\cite{Dawei2004},           pseudo          random          number
generation~\cite{stojanovski2001chaos}.
All  these  applications  capitalize   fundamental  properties  of  the  chaotic
maps~\cite{Devaney},  namely:  sensitive   dependence  on   initial  conditions,
transitivity, and density of periodic points.
A system  is sensitive  to initial conditions  if any point contains,
in any neighborhood, another point  with a completely different future trajectory.
Topological  transitivity  is  established  when,  for  any  element,  any
neighborhood of its future evolution eventually 
 overlaps  with any other open
set.   On the  contrary,  a  dense set  of  periodic points  is  an element  of
regularity that a chaotic dynamical system has to exhibit.

Chaotic  discrete-time dynamical  systems  are iterative  processes
  defined by  a
chaotic map $f$  from a domain $E$ to itself.   Starting from any configurations
$x \in E$, the system produces the sequence $x,f(x),f^2(x), f^3(x),\dots$, where
$f^k(x)$  is the  $k$-th iterate  of  $f$ at  $x$.  Works  referenced above  are
instances of that scheme: they iterate \emph{tent} or \emph{logistic} maps known
to be chaotic on $\R$.

As far as we  know, no result so far states that
the chaotic properties  of a function that
has  been  theoretically proven  on  $\R$  remain  valid on  the  floating-point
numbers, which is the implementation domain.   Thus, to avoid the 
loss of chaos this
work presents an alternative: to construct, from Boolean networks $f : \Bool^n \rightarrow \Bool^n$, continuous fonctions $G_f$ defined on the domain   $\llbracket 1 ;  n \rrbracket^{\Nats}  \times \Bool^n$,
where $\llbracket 1  ; n \rrbracket$ is the interval of integers
$\{1, 2, \hdots,  n\}$ and $\Bool$ is the Boolean domain  $\{0,1\}$.  Due to the
discrete nature of $f$, theoretical results obtained on $G_f$ are preserved in
implementation

Furthermore,  instead  of  finding  an example of such maps  and  to  prove  the
chaoticity  of  their  discrete-time   iterations,  we  tackle  the  problem  of
characterizing all the maps with chaotic iterations according to Devaney's chaos
definition~\cite{Devaney}.    This   is    the    first   contribution.     This
characterization is  expressed on  the asynchronous iteration graph  of the Boolean  map $f$,
which   contains  $2^n$  vertices.    To  extend   the  applicability   of  this
characterization, sufficient conditions  that ensure this  chaoticity are expressed  on the
interaction graph of $f$, which only  contains $n$ vertices.  This is the second
contribution.  Starting thus with  an interaction graph with required properties,
all the  maps resulting from a Boolean network
constructed on this graph have chaotic iterations.
Eventually, the approach is applied on a pseudo random number generation (PRNG).
Uniform distribution of the output, which  is a necessary condition for PRNGs is
then addressed. Functions with such property are thus characterized again on the
asynchronous iteration graph.  This is the third contribution.  The relevance of the approach
and the application to pseudo random number generation are evaluated thanks to a
classical test suite.

The  rest of  the paper  is organized  as  follows.  Section~\ref{section:chaos}
recalls   discrete-time  Boolean   dynamical  systems.    Their   chaoticity  is
characterized  in  Sect.~\ref{sec:charac}.    Sufficient  conditions  to  obtain
chaoticity  are presented  in  Sect.~\ref{sec:sccg}. The  application to  pseudo
random  number   generation  is  formalized,   maps  with  uniform   output  are
characterized, and PRNGs are  evaluated in Sect.~\ref{sec:prng}.  The paper ends
with  a conclusion  section where intended
future work is presented.

\section{Preliminaries}
\label{section:chaos}

Let $n$ be a positive integer. A Boolean network is a discrete dynamical
system defined from a {\emph{Boolean map}}
\[
f:\Bool^n\to\Bool^n,\qquad x=(x_1,\dots,x_n)\mapsto f(x)=(f_1(x),\dots,f_n(x)),
\]
and an {\emph{iteration scheme}} (\textit{e.g.}, parallel, sequential,
asynchronous\ldots). 
For instance, with the parallel iteration scheme, 
given an initial configuration $x^0\in\Bool^n$, the dynamics
of the system are described by the recurrence $x^{t+1}=f(x^t)$.
The retained scheme only modifies one element at 
each iteration and is further referred by \emph{asynchronous}. In other words, at the $t^{th}$ iteration, only the $s_{t}-$th component is
``iterated'', where $s = \left(s_t\right)_{t \in \mathds{N}}$ is a sequence of indices taken in $\llbracket 1;n \rrbracket$ called ``strategy''. Formally,
let $F_f: \llbracket1;n\rrbracket\times \Bool^{n}$ to $\Bool^n$ be defined by
\[
F_f(i,x)=(x_1,\dots,x_{i-1},f_i(x),x_{i+1},\dots,x_n).
\]
With the asynchronous iteration  scheme, given an initial configuration
$x^0\in\Bool^n$ and a strategy $s\in
\llbracket1;n\rrbracket^\Nats$, the dynamics of the network
are described by the recurrence
\begin{equation}\label{eq:asyn}
x^{t+1}=F_f(s_t,x^t).
\end{equation}
Let $G_f$ be the map from $\llbracket1;n\rrbracket^\Nats\times\Bool^n$ to itself defined by
\[
G_f(s,x)=(\sigma(s),F_f(s_0,x)),
\] 
where $\forall t\in\Nats,\sigma(s)_t=s_{t+1}$. 
The parallel iteration of $G_f$ from an initial point
$X^0=(s,x^0)$ describes the ``same dynamics'' as the asynchronous
iteration of $f$ induced by $x^0$ and the strategy
$s$ (this is why $G_f$ has been introduced). 

Consider the space $\mathcal{X}=\llbracket 1;n\rrbracket^{\Nats}\times
\Bool^n$. The distance $d$ between two points $X=(s,x)$ and
$X'=(s',x')$ in $\mathcal{X}$ is defined by
\[
d(X,X')= d_H(x,x')+d_S(s,s'),~\textrm{where}~
\left\{
\begin{array}{l}
\displaystyle{d_H(x,x')=\sum_{i=1}^n |x_i-x'_i|}\\[5mm] 
\displaystyle{d_S(s,s')=\frac{9}{n}\sum_{t\in\Nats}\frac{|s_t-s'_t|}{10^{t+1}}}.
\end{array}
\right.
\]
Thus, $\lfloor d(X,X')\rfloor=d_H(x,x')$ is the Hamming distance
between $x$ and $x'$, and $d(X,X')-\lfloor d(X,X')\rfloor=d_S(s,s')$ 
measures the differences between $s$ and $s'$. More
precisely, this floating part is less than $10^{-k}$ if and only if
the first $k$ terms of the two strategies are equal. Moreover, if the
$k^{th}$ digit is nonzero, then $s_k\neq s'_k$. 

Let $f$ be any map from $\Bool^n$ to itself, and $\neg:\Bool^n \rightarrow \Bool^n$ defined by
$\neg(x)=(\overline{x_1},\dots,\overline{x_n})$. Considering this
distance $d$ on $\mathcal{X}$, it has already been proven that
\cite{guyeux10}:
\begin{itemize}
\item 
$G_f$ is \emph{continuous},
\item 
the parallel iteration of $G_\neg$ is \emph{regular} (periodic points of $G_{\neg}$ are dense in $\mathcal{X}$),
\item 
$G_{\neg}$ is \emph{topologically transitive} (for all $X,Y \in
\mathcal{X}$, and for all open balls $B_X$ and $B_Y$ centered in $X$ and
$Y$ respectively, there exist  $X'\in B_X$ and $t \in
\mathds{N}$ such that $G_{\neg}^t(X') \in B_Y$),
\item $G_{\neg}$ has 
\emph{sensitive dependence on initial conditions}
(there exists $\delta >0$ such that for any $X\in \mathcal{X}$
and any open ball $B_X$, there exist $X'\in B_X$ and $t\in\Nats$
 such that $d(G_{\neg}^t(X), G_{\neg}^t(X'))>\delta $).
\end{itemize}
Particularly, $G_{\neg}$ is {\emph{chaotic}}, according to the Devaney's
definition recalled below:

\begin{definition}[Devaney~\cite{Devaney}]
A continuous map $f$ on a metric space $(\mathcal{X},d)$ is chaotic 
if it is regular, sensitive, and topologically transitive.
\end{definition}

In other words, quoting Devaney in~\cite{Devaney}, a chaotic dynamical system ``is unpredictable because of the sensitive dependence on initial conditions. It cannot be broken down or simplified into two subsystems which do not interact because of topological transitivity. And in the midst of this random behavior, we nevertheless have an element of regularity''. Let us finally remark that the definition above is redundant: Banks \emph{et al.} have proven that sensitivity is indeed implied by regularity and transitivity \cite{Banks92}.

\section{Characterization of Chaotic Discrete-Time Dynamical Systems}\label{sec:charac}
In this section, we give a characterization of Boolean networks $f$ 
making the iterations of any induced map $G_f$ chaotic.
This is achieved by establishing inclusion relations 
between the transitive, regular, and chaotic sets defined below:
\begin{itemize}
\item   $\mathcal{T}   =    \left\{f   :   \mathds{B}^n   \to
\mathds{B}^n \big/ G_f \textrm{ is transitive} \right\}$,
\item   $\mathcal{R}   =    \left\{f   :   \mathds{B}^n   \to
\mathds{B}^n \big/ G_f \textrm{ is regular} \right\}$,
\item   $\mathcal{C}   =    \left\{f   :   \mathds{B}^n   \to
\mathds{B}^n  \big/  G_f  \textrm{  is chaotic  (Devaney)} \right\}$.
\end{itemize}

Let $f$ be a map from $\Bool^n$ to itself. The
{\emph{asynchronous iteration graph}} associated with $f$ is the
directed graph $\Gamma(f)$ defined by: the set of vertices is
$\Bool^n$; for all $x\in\Bool^n$ and $i\in \llbracket1;n\rrbracket$,
the graph $\Gamma(f)$ contains an arc from $x$ to $F_f(i,x)$. 
The relation between $\Gamma(f)$ and $G_f$ is clear: there exists a
path from $x$ to $x'$ in $\Gamma(f)$ if and only if there exists a
strategy $s$ such that the parallel iteration of $G_f$ from the
initial point $(s,x)$ reaches the point $x'$.
Finally, in what follows the term \emph{iteration graph} is a shortcut for 
asynchronous iteration graph.

We can thus characterize $\mathcal{T}$:

\begin{proposition} $G_f$  is transitive if  and only if
 $\Gamma(f)$ is strongly connected.
\end{proposition}

\begin{proof} 

$\Longleftarrow$ Suppose that $\Gamma(f)$ is strongly connected. Let
$(s,x)$ and $(s',x')$ be two points of $\mathcal{X}$, and let
$\varepsilon >0$. We will define a strategy $\tilde s$ such that the
distance between $(\tilde s,x)$ and $(s,x)$ is less than
$\varepsilon$, and such that the parallel iterations of $G_f$ from
$(\tilde s,x)$ reaches the point $(s',x')$.

Let $t_1 =\lfloor-\log_{10}(\varepsilon)\rfloor$, and let $x''$ be the
configuration of $\Bool^n$ that we obtain from $(s,x)$ after $t_1$ iterations
of $G_f$. Since $\Gamma(f)$ is strongly connected, there exists a
strategy $s''$ and $t_2\in\Nats$ such that, $x'$ is reached from
$(s'',x'')$ after $t_2$ iterations~of~$G_f$.

Now, consider the strategy $\tilde
s=(s_0,\dots,s_{t_1-1},s''_0,\dots,s''_{t_2-1},s'_0,s'_1,s'_2,s'_3\dots)$.
It is clear that $(s',x')$ is reached from $(\tilde s,x)$ after
$t_1+t_2$ iterations of $G_f$, and since $\tilde s_t=s_t$ for $t<t_1$,
by the choice of $t_1$, we have $d((s,x),(\tilde
s,x))<\epsilon$. Consequently, $G_f$ is transitive.

$\Longrightarrow$ If $\Gamma(f)$ is not strongly connected, then 
there exist two configurations $x$ and $x'$ such that $\Gamma(f)$ has no path
from $x$ to $x'$. Let $s$ and $s'$ be two strategies, and let
$0<\varepsilon<1$. Then, for all $(s'',x'')$ such that
$d((s'',x''),(s,x))<\varepsilon$, we have $x''=x$, so that iteration
of $G_f$ from $(s'',x'')$ only reaches points in $\mathcal{X}$ that are at
a greater distance than one with $(s',x')$. So $G_f$ is not transitive.
\end{proof}




We now prove that:

\begin{proposition}
\label{Prop: T est dans R} $\mathcal{T} \subset \mathcal{R}$.
\end{proposition}





\begin{proof}  
Let $f:\Bool^n\to\Bool^n$ such that $G_f$ is transitive ($f$ is in
$\mathcal{T}$). Let $(s,x) \in\mathcal{X}$ and $\varepsilon >0$. To
prove that $f$ is in $\mathcal{R}$, it is sufficient to prove that
there exists a strategy $\tilde s$ such that the distance between
$(\tilde s,x)$ and $(s,x)$ is less than $\varepsilon$, and such that
$(\tilde s,x)$ is a periodic point.

Let $t_1=\lfloor-\log_{10}(\varepsilon)\rfloor$, and let $x'$ be the
configuration that we obtain from $(s,x)$ after $t_1$ iterations of
$G_f$. According to the previous proposition, $\Gamma(f)$ is strongly
connected. Thus, there exists a strategy $s'$ and $t_2\in\Nats$ such
that $x$ is reached from $(s',x')$ after $t_2$ iterations of $G_f$.

Consider the strategy $\tilde s$ that alternates the first $t_1$ terms
of $s$ and the first $t_2$ terms of $s'$: $\tilde
s=(s_0,\dots,s_{t_1-1},s'_0,\dots,s'_{t_2-1},s_0,\dots,s_{t_1-1},s'_0,\dots,s'_{t_2-1},s_0,\dots)$. It
is clear that $(\tilde s,x)$ is obtained from $(\tilde s,x)$ after
$t_1+t_2$ iterations of $G_f$. So $(\tilde s,x)$ is a periodic
point. Since $\tilde s_t=s_t$ for $t<t_1$, by the choice of $t_1$, we
have $d((s,x),(\tilde s,x))<\epsilon$.
\end{proof}

\begin{remark} 
Inclusion of proposition \ref{Prop: T est dans
R}  is strict,  due to  the identity  map (which  is regular,  but not
transitive).
\end{remark}

We can thus conclude  that $\mathcal{C} = \mathcal{R} \cap \mathcal{T}
= \mathcal{T}$, which leads to the following characterization:

\begin{theorem}
\label{Th:Caractérisation   des   IC   chaotiques}  
Let $f:\Bool^n\to\Bool^n$. $G_f$ is chaotic  (according to  Devaney) 
if and only if $\Gamma(f)$ is strongly connected.
\end{theorem}

\section{Generating Strongly Connected Iteration Graph}\label{sec:sccg}
The previous  section  has  shown  the  interest  of  strongly  connected  iteration
graphs.  This section  presents two  approaches to  generate functions with such
property. The  first is  algorithmic (Sect.~\ref{sub:sccg2}) whereas  the second
gives   a   sufficient   condition   on the interaction graph of 
the Boolean map $f$   to   get  a strongly connected iteration graph
(Sect.~\ref{sub:sccg1}).
  
\subsection{Algorithmic       Generation       of       Strongly       Connected
Graphs}\label{sub:sccg2} 

This section presents a first solution to compute a map $f$ with a
strongly connected graph of iterations $\Gamma(f)$.  It is based on a
generate and test approach.

We first consider the negation function $\neg$ whose iteration graph
$\Gamma(\neg)$ is obviously strongly connected. 
Given a graph $\Gamma$, initialized with
$\Gamma(\neg)$, the algorithm iteratively does the two following
stages:
\begin{enumerate}
\item randomly select an edge of the current iteration graph $\Gamma$ and
\item  check whether  the  current  iteration graph  without  that edge  remains
strongly connected (by a  Tarjan algorithm~\cite{Tarjanscc72}, for instance). In
the positive case the edge is removed from $\Gamma$,
\end{enumerate} until  a rate $r$ of  removed edges is greater  than a threshold
given by the user.
If $r$  is close to $0\%$ (\textit{i.e.}, few  edges are removed), there
should remain about $n\times 2^n$ edges.   In the opposite case, if $r$ is close
to $100\%$, there are about $2^n$ edges left.  
In all cases, this step returns
the last graph  $\Gamma$ that is strongly connected.  It is  now then obvious to
return the function $f$ whose iteration graph is $\Gamma$.

Even if this algorithm always returns functions with stroncly connected component (SCC) iteration graph, 
it suffers from 
iteratively verifying connexity on 
the whole iteration graph, \textit{i.e.},
on a graph with $2^n$ vertices.   
Next section tackles this problem: it presents sufficient conditions
on a graph reduced to $n$ elements     
that allow to obtain SCC iteration graph.

\subsection{Sufficient Conditions to Strongly Connected Graph}\label{sub:sccg1}

We are looking for maps $f$ such that interactions between $x_i$ and
$f_j$ make its iteration graph $\Gamma(f)$ strongly connected.
We first need additional notations and
definitions. For $x\in\Bool^n$ and $i\in\llbracket 1;n\rrbracket$, we
denote by $\overline{x}^i$ the configuration that we obtain be switching the
$i-$th component of $x$, that is,
$\overline{x}^i=(x_1,\dots,\overline{x_i},\dots,x_n)$. Information
interactions between the components of the
system are obtained from the {\emph{discrete Jacobian matrix}} $f'$
of $f$, which is defined as being the map  which associates to each configuration
$x\in\Bool^n$, the $n\times n$ matrix
\[
f'(x)=(f_{ij}(x)),\qquad 
f_{ij}(x)=\frac{f_i(\overline{x}^j)-f_i(x)}{\overline{x}^j_j-x_j}\qquad (i,j\in\llbracket1;n\rrbracket).
\]

More precisely, interactions are represented under the form of a
signed directed graph $G(f)$ defined by: the set of vertices is
$\llbracket1;n\rrbracket$, and there exists an arc from $j$ to $i$ of
sign $s\in\{-1,1\}$, denoted $(j,s,i)$, if $f_{ij}(x)=s$ for at least
one $x\in\Bool^n$. Note that the presence of both a positive and a
negative arc from one vertex to another is allowed.

Let $P$ be a sequence of arcs of $G(f)$ of the form
\[
(i_1,s_1,i_2),(i_2,s_2,i_3),\ldots,(i_r,s_r,i_{r+1}).
\]
Then, $P$ is said to be a path of $G(f)$ of length $r$ and of sign
$\Pi_{i=1}^{r}s_i$, and $i_{r+1}$ is said to be reachable from
$i_1$. $P$ is a {\emph{circuit}} if $i_{r+1}=i_1$ and if the vertices
$i_1$,\ldots $i_r$ are pairwise distinct. A vertex $i$ of $G(f)$ has a
positive (resp. negative) {\emph{loop}}, if $G(f)$ has a positive
(resp. negative) arc from $i$ to itself.

Let $\alpha\in\Bool$. We denote by $f^{\alpha}$ the map from $\Bool^{n-1}$ 
to itself defined for
any $x\in\Bool^{n-1}$ by 
\[
f^{\alpha}(x)=(f_1(x,\alpha),\dots,f_{n-1}(x,\alpha)).
\]
We denote by $\Gamma(f)^\alpha$ the
subgraph of $\Gamma(f)$ induced by the subset
$\Bool^{n-1} \times \{\alpha\}$ of $\Bool^n$.
Let us give and prove the following technical lemma:

\begin{lemma}\label{lemma:subgraph}
$G(f^\alpha)$ is a subgraph of $G(f)$: every arc of $G(f^\alpha)$ is
an arc of $G(f)$. Furthermore, if $G(f)$ has no arc from $n$ to a
vertex $i\neq n$, then $G(f^\alpha)=G(f)\setminus n$: one obtains
$G(f^\alpha)$ from $G(f)$ by removing vertex $n$ as well as all the
arcs with $n$ as initial or final vertex.
\end{lemma}

\begin{proof}
Suppose that $G(f^{\alpha})$ has an arc from $j$ to $i$ of sign
$s$. By definition, there exists $x\in\Bool^{n-1}$ such that
$f^{\alpha}_{ij}(x)=s$, and since it is clear that
$f^{\alpha}_{ij}(x)=f_{ij}(x,\alpha)$, we deduce that $G(f)$ has an
arc from $j$ to $i$ of sign $s$. This proves the first assertion. 
To demonstrate the second assertion, it is sufficient to prove that if
$G(f)$ has an arc from $i$ to $j$ of sign $s$, with $i,j\neq n$, then
$G(f^\alpha)$ also contains this arc. So suppose that $G(f)$ has an
arc from $i$ to $j$ of sign $s$, with $i,j\neq n$. Then, there exists
$x\in\Bool^{n-1}$ and $\beta\in\Bool$ such that
$f_{ij}(x,\beta)=s$. If $f_{ij}(x,\beta)\neq f_{ij}(x,\alpha)$, then
$f_i$ depends on the $n-$th component, in
contradiction with the assumptions. So $f_{ij}(x,\alpha)=s$. It is
then clear that $f^{\alpha}_{ij}(x)=s$, that is, $G(f^\alpha)$ has an
arc from $j$ to $i$ of sign $s$.
\end{proof}

\begin{lemma}\label{lemma:iso}
$\Gamma(f^\alpha)$ and $\Gamma(f)^\alpha$ are isomorphic.
\end{lemma}

\begin{proof}
Let $h$ be the bijection from $\Bool^{n-1}$ to
$\Bool^{n-1}\times \{\alpha\}$ defined by $h(x)=(x,\alpha)$ for all
$x\in\Bool^{n-1}$.
It is easy to see that $h$ is an isomorphism
between $\Gamma(f^\alpha)$ and $\Gamma(f)^\alpha$ that is: $\Gamma(f^\alpha)$ has an
arc from $x$ to $y$ if and only if $\Gamma(f)^\alpha$ has an arc from
$h(x)$ to $h(y)$.
\end{proof}

\begin{theorem}\label{th:Adrien}
Let $f$ be a map from $\Bool^n$ to itself such that:
\begin{enumerate}
\item 
$G(f)$ has no cycle of length at least two;
\item 
every vertex of $G(f)$ with a positive loop has also a negative loop;
\item
every vertex of $G(f)$ is reachable from a vertex with a negative loop.
\end{enumerate}
Then, $\Gamma(f)$ is strongly connected. 
\end{theorem}

\begin{proof}
By induction on $n$. Let $f$ be a map from $\Bool^n$ to itself
satisfying the conditions of the statement. If $n=1$ the result is
obvious: according to the third point of the statement, $G(f)$ has a
negative loop; so $f(x)=\overline{x}$ and $\Gamma(f)$ is a cycle of
length two. Assume that $n>1$ and that the theorem is valid for maps
from $\Bool^{n-1}$ to itself. According to the first point of the
statement, $G(f)$ contains at least one vertex $i$ such that $G(f)$
has no arc from $i$ to a vertex $j\neq i$. Without loss of generality,
assume that $n$ is such a vertex. Then, according to
Lemma~\ref{lemma:subgraph}, $f^0$ and $f^1$ satisfy the conditions of
the statement. So, by induction hypothesis, $\Gamma(f^0)$ and
$\Gamma(f^1)$ are strongly connected. So, according to
Lemma~\ref{lemma:iso}, $\Gamma(f)^0$ and $\Gamma(f)^1$ are strongly
connected. To prove that $\Gamma(f)$ is strongly connected, it is
sufficient to prove that $\Gamma(f)$ contains an arc $x\to y$ with
$x_n=0<y_n$ and an arc $x\to y$ with $x_n=1>y_n$. In other words, it
is sufficient to prove that:
\begin{equation}\tag{$*$}
\forall \alpha\in\Bool,~\exists x\in\Bool^n,\qquad  x_n=\alpha\neq f_n(x).
\end{equation}

Assume first that $n$ has a negative loop. Then, by the definition of
$G(f)$, there exists $x\in\Bool^n$ such that $f_{nn}(x)<0$. Consequently,
if $x_n=0$, we have $f_n(x)>f_n(\overline{x}^n)$, so $x_n=0\neq f_n(x)$ and
$\overline{x}^n_n=1\neq f_n(\overline{x}^n)$; and if $x_n=1$, we have
$f_n(x)<f_n(\overline{x}^n)$, so $x_n=1\neq f_n(x)$ and $\overline{x}^n_n=0\neq
f_n(\overline{x}^n)$. In both cases, the condition ($*$) holds.

Now, assume that $n$ has no negative loop. According to the second
point of the statement, $n$ has no loop, \emph{i.e.}, the value of $f_n(x)$
does not depend on the value of $x_n$. According to the third point of
the statement, $n$ is not of in-degree zero in $G(f)$, \emph{i.e.}, $f_n$ is
not a constant. Consequently, there exists $x,y\in \Bool^n$ such that
$f_n(x)=1$ and $f_n(y)=0$. Let $x'=(x_1,\dots,x_{n-1},0)$ and
$y'=(y_1,\dots,y_{n-1},1)$. Since the value of $f_n(x)$
(resp. $f_n(y)$) does not depend on the value of $x_n$ (resp. $y_n$),
we have $f_n(x')=f_n(x)=1\neq x'_n$ (resp. $f_n(y')=f_n(y)=0\neq
y'_n$). So the condition ($*$) holds, and the theorem is proven.
\end{proof}





\section{Application to Pseudo Random Number  Generator}\label{sec:prng}
This section presents a direct application of the theory
developed above.

\subsection{Boolean and Chaos based PRNG}

We have proposed in~\cite{bgw09:ip} a new family of generators that receives 
two PRNGs as inputs. These two generators are mixed with chaotic iterations, 
leading thus to a new PRNG that improves the statistical properties of each
generator taken alone. Furthermore, our generator 
possesses various chaos properties
that none of the generators used as input present.
This former family of PRNGs was reduced to chaotic iterations
of the negation function, \textit{i.e.}, reduced to $G_\neg$. 
However, it is possible to use any function $f$ such that  $G_f$ is chaotic (s.t. the graph $\Gamma(f)$ is strongly connected).

\begin{algorithm}[t]
\KwIn{a function $f$, an iteration number $b$, an initial configuration $x^0$ ($n$ bits)}
\KwOut{a configuration $x$ ($n$ bits)}
$x\leftarrow x^0$\;
$k\leftarrow b + \textit{XORshift}(b+1)$\;
\For{$i=0,\dots,k-1$}
{
$s\leftarrow{\textit{XORshift}(n)}$\;
$x\leftarrow{F_f(s,x)}$\;
}
return $x$\;
\caption{PRNG with chaotic functions}
\label{CI Algorithm}
\end{algorithm}

\begin{algorithm}
\SetAlgoLined
\KwIn{the internal configuration $z$ (a 32-bit word)}
\KwOut{$y$ (a 32-bit word)}
$z\leftarrow{z\oplus{(z\ll13)}}$\;
$z\leftarrow{z\oplus{(z\gg17)}}$\;
$z\leftarrow{z\oplus{(z\ll5)}}$\;
$y\leftarrow{z}$\;
return $y$\;
\medskip
\caption{An arbitrary round of \textit{XORshift} algorithm}
\label{XORshift}
\end{algorithm}

This generator is synthesized in Algorithm~\ref{CI Algorithm}.
It takes as input: a function $f$;
an integer $b$, ensuring that the number of executed iterations is at least $b$ and at most $2b+1$; and an initial configuration $x^0$.
It returns the new generated configuration $x$.  Internally, it embeds two
\textit{XORshift}$(k)$ PRNGs \cite{Marsaglia2003} that returns integers uniformly distributed
into $\llbracket 1 ; k \rrbracket$.
\textit{XORshift} is a category of very fast PRNGs designed by George Marsaglia, which repeatedly uses the transform of exclusive or (XOR, $\oplus$) on a number with a bit shifted version of it. This PRNG, which has a period of $2^{32}-1=4.29\times10^9$, is summed up in Algorithm~\ref{XORshift}. It is used in our PRNG to compute the strategy length and the strategy elements.

We are then left to instantiate the function $f$ in 
Algorithm~\ref{CI Algorithm} according to approaches detailed
in Sect.~\ref{sec:sccg}.
Next section shows how the 
uniformity of distribution has been taken into account.

\subsection{Uniform Distribution of the Output}  

Let us firstly recall that a stochastic matrix is a square matrix where all entries are nonnegative and all rows sum to 1. A double stochastic matrix is a stochastic matrix where all columns sum to 1. Finally, a stochastic matrix $M$ of size $n$ 
is regular if $\exists k \in \mathds{N}^\ast, \forall i,j \in \llbracket 1; n \rrbracket, M_{ij}^k>0$.
 The following theorem is well-known:
\begin{theorem}\label{th}
  If $M$ is a regular stochastic matrix, then $M$ 
  has an unique stationary  probability vector $\pi$. Moreover, 
  if $\pi^0$ is any initial probability vector and 
  $\pi^{k+1} = \pi^k.M $ for $k = 0, 1,\dots$ then the Markov chain $\pi^k$
  converges to $\pi$ as $k$ tends to infinity.
\end{theorem}

\begin{figure}[t]
  \begin{center}
    \subfloat[$\Gamma(g)$.]{
      \begin{minipage}{0.13\textwidth}
        \begin{center}
          \includegraphics[height=3cm]{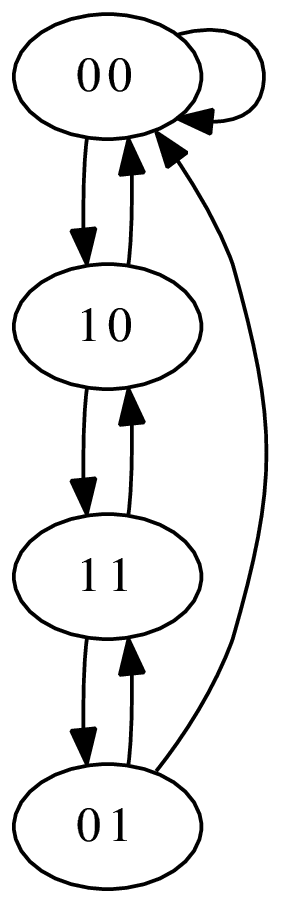}
        \end{center}
      \end{minipage}
      \label{fig:g:iter}
    }
    \subfloat[$G(g)$.]{
      \begin{minipage}{0.13\textwidth}
        \begin{center}
         \includegraphics[height=1.5cm]{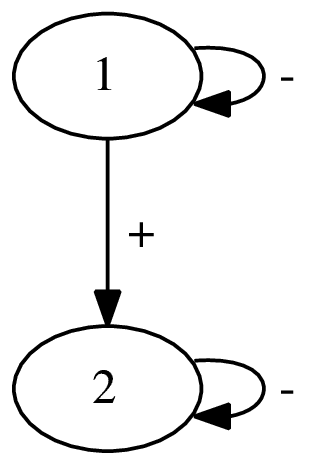}
        \end{center}
      \end{minipage}
      \label{fig:g:inter}
    }
    \subfloat[$\check{M}_g $.]{
      \begin{minipage}{0.13\textwidth}
        \begin{center}
          $\left( 
            \begin{array}{cccc} 
              1 & 0 & 1 & 0 \\ 
              1 & 0 & 0 & 1 \\ 
              1 & 0 & 0 & 1 \\ 
              0 & 1 & 1 & 0 
            \end{array}
          \right)
          $
        \end{center}
      \end{minipage}
      \label{fig:g:incidence}
    }
    \subfloat[$\Gamma(h)$.]{
      \begin{minipage}{0.13\textwidth}
        \begin{center}
          \includegraphics[height=3cm]{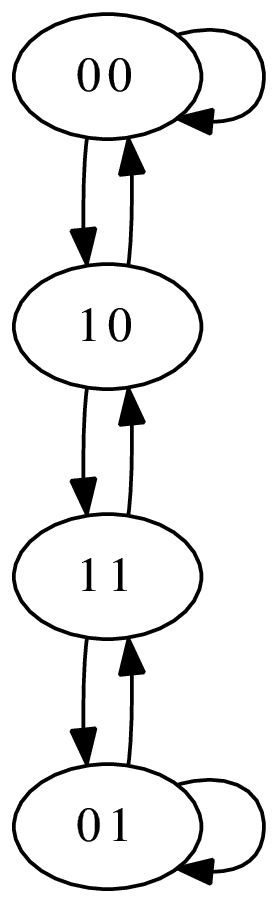}
        \end{center}
      \end{minipage}
      \label{fig:h:iter}
    }
      \subfloat[$G(h)$.]{
        \begin{minipage}{0.13\textwidth}
          \begin{center}
            \includegraphics[height=1.5cm]{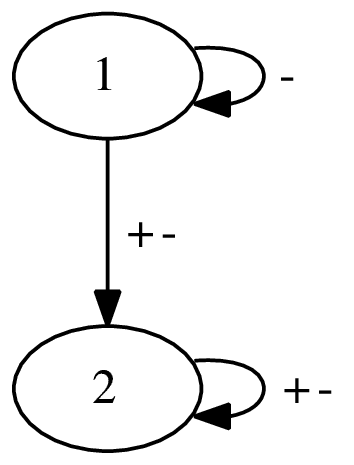}
          \end{center}
        \end{minipage}
        \label{fig:h:inter}
      }
      \subfloat[$\check{M}_h $.]{
        \begin{minipage}{0.13\textwidth}
          \begin{center}
            $\left( 
              \begin{array}{cccc} 
                1 & 0 & 1 & 0 \\ 
                0 & 1 & 0 & 1 \\ 
                1 & 0 & 0 & 1 \\ 
                0 & 1 & 1 & 0 
              \end{array}
            \right)
            $
          \end{center}
        \end{minipage}
        \label{fig:h:incidence}
      }
    \end{center}
    \caption{Graphs of candidate functions  with $n=2$}
    \label{fig:xplgraph}
  \end{figure}

Let us explain on a small example with 2 elements 
that the application of such a theorem allows to verify whether 
the output is uniformly distributed or not.
Let then $g$ and $h$ be the two functions  from $\Bool^2$
to itself defined in Fig.~\ref{fig:xplgraph}
and whose iteration graphs are strongly connected.
As the \textit{XORshift} PRNG is uniformly distributed, 
the strategy is uniform on $\llbracket 1, 2 \rrbracket$,
and each edge of $\Gamma(g)$ and of $\Gamma(h)$ 
has a probability $1/2$ to be traversed.
In other words, $\Gamma(g)$ is the oriented graph of a Markov chain.
It is thus easy to verify that the transition matrix of such a process
is $M_g   = \frac{1}{2} \check{M}_g$, 
where $\check{M}_g$ is the adjacency matrix given in 
Fig.~\ref{fig:g:incidence}, and  similarly for $M_h$.

Both $M_g$ and $M_h$ are (stochastic and) regular since no element is null either in
$M^4_g$ or in $M^4_h$.
Furthermore, the probability vectors $\pi_g=(0.4, 0.1, 0.3,0.2)$ and
$\pi_h=(0.25,0.25,0.25,0.25)$ verify $\pi_g M_g = \pi_g$ and 
$\pi_h M_h = \pi_h$. 
Thus, due to Theorem \ref{th}, 
for any initial probability vector $\pi^0$, we have 
$\lim_{k \to \infty} \pi^0 M^k_g = \pi_g$ and 
$\lim_{k \to \infty} \pi^0 M^k_h = \pi_h$. 
So the Markov process associated to $h$ tends to the uniform 
distribution whereas the one associated to $g$ does not. 
It induces that $g$ shouldn't be iterated in a PRNG.
On the contrary, 
$h$ can be embedded into the PRNG Algorithm~\ref{CI Algorithm}, 
provided the number $b$ of iterations between two successive
values is sufficiently large so that the Markov process becomes 
close to the uniform distribution.

Let us first prove the following technical lemma.
\begin{lemma}\label{lem:stoc}
Let $f: \Bool^{n} \rightarrow \Bool^{n}$, $\Gamma(f)$ its iteration graph, $\check{M}$ the adjacency
matrix of $\Gamma(f)$, and $M$ a $n\times n$ matrix defined by 
$
M_{ij} = \frac{1}{n}\check{M}_{ij}$ 
if $i \neq j$ and  
$M_{ii} = 1 - \frac{1}{n} \sum\limits_{j=1, j\neq i}^n \check{M}_{ij}$ otherwise.
Then $M$ is a regular stochastic matrix iff $\Gamma(f)$ is strongly connected.
\end{lemma}

\begin{proof}  Notice first that $M$ is a stochastic matrix by construction.
If there exists $k$ s.t. $M_{ij}^k>0$ for any $i,j\in \llbracket
1;  2^n  \rrbracket$, the  inequality  $\check{M}_{ij}^k>0$  is thus  established.
Since $\check{M}_{ij}^k$ is the number of  paths from $i$ to $j$ of length $k$
in $\Gamma(f)$ and since such a number is positive, thus
$\Gamma(f)$ is strongly connected.

Conversely, if $\Gamma(f)$ is SCC, then for all vertices $i$ and $j$, a path can
be  found to  reach $j$  from $i$  in at  most $2^n$  steps.  There  exists thus
$k_{ij} \in \llbracket 1,  2^n \rrbracket$ s.t. $\check{M}_{ij}^{k_{ij}}>0$.  
As all the multiples $l \times k_{ij}$ of $k_{ij}$ are such that 
$\check{M}_{ij}^{l\times  k_{ij}}>0$, 
we can  conclude that, if
$k$ is the least common multiple of $\{k_{ij}  \big/ i,j  \in \llbracket 1,  2^n \rrbracket  \}$ thus 
$\forall i,j  \in \llbracket  1, 2^n \rrbracket,  \check{M}_{ij}^{k}>0$. 
So, $\check{M}$ and thus $M$ are regular.
\end{proof}


With such a material, we can formulate and prove the following theorem.
\begin{theorem}
  Let $f: \Bool^{n} \rightarrow \Bool^{n}$, $\Gamma(f)$ its
  iteration graph, $\check{M}$ its adjacency
  matrix and $M$ a $n\times n$ matrix defined as in the previous lemma.
  If $\Gamma(f)$ is SCC then 
  the output of the PRNG detailed in Algorithm~\ref{CI Algorithm} follows 
  a law that tends to the uniform distribution 
  if and only if $M$ is a double stochastic matrix.
\end{theorem} 
\begin{proof}
$M$ is a regular stochastic matrix (Lemma~\ref{lem:stoc}) that has a unique stationary probability vector (Theorem \ref{th}).
  Let $\pi$  be 
  $\left(\frac{1}{2^n}, \hdots, \frac{1}{2^n} \right)$.
  We have $\pi M = \pi$ iff 
  the sum of values of each column of $M$  is one,  
  \textit{i.e.}, iff
  $M$ is double stochastic.
  \end{proof}

\subsection{Experiments}

Let us consider the interaction graph $G(f)$ given in Fig.~\ref{fig:G}. 
It verifies Theorem~\ref{th:Adrien}: 
all the functions $f$ whose interaction graph is $G(f)$ have then 
a strongly connected iteration graph $\Gamma(f)$.
Practically, a simple constraint solving has found 
520 non isomorphic functions 
and only 16 of them have a double stochastic matrix.
Figure~\ref{fig:listfonction} synthesizes them by 
defining the images of 
0,1,2,\ldots,14,15. 
Let $e_j$ be the unit vector in the canonical basis, 
the third column gives 
$$
\max\limits_{j \in \llbracket 1, 2^n \rrbracket} 
\{
\min \{
 k \mid k \in \Nats, \vectornorm{\pi_j M_f^k - \pi} < 10^{-4}
\}
\}
\textrm{where $\pi_j$ is $1/n. e_j$,} 
$$
that is the smallest iteration number 
that is sufficient to obtain a deviation less than $10^{-4}$
from the uniform distribution. 
Such a number is the parameter $b$ in Algorithm~\ref{CI Algorithm}.

\begin{figure}[th]
  \begin{center}
  \subfloat[Interaction Graph]{
    \begin{minipage}{0.20\textwidth}
      \begin{center}
        \includegraphics[width=3.5cm]{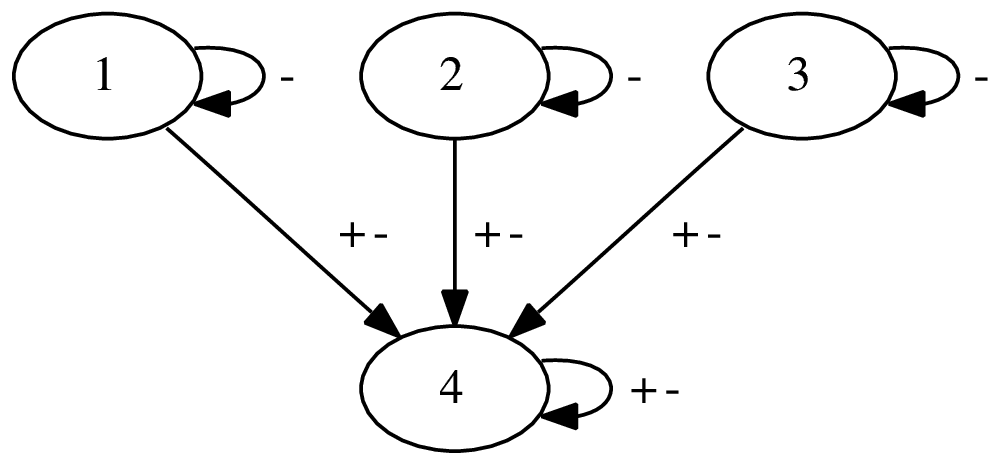}
      \end{center}
    \end{minipage}
    \label{fig:G}
  }\hfill
  \subfloat[Double Stochastic Functions]{
    \begin{minipage}{0.75\textwidth}
      \begin{scriptsize}
        \begin{center}
          \begin{tabular}{|c|c|c|}
\hline
{Name}& {Function image}&{$b$} \\
\hline 
$\mathcal{F}_1$ & 14, 15, 12, 13, 10, 11, 8, 9, 6, 7, 4, 5, 2, 3, 1, 0  & 206\\
\hline
$\mathcal{F}_2$ &14, 15, 12, 13, 10, 11, 8, 9, 6, 7, 5, 4, 3, 2, 0, 1  
 & 94 \\
\hline
$\mathcal{F}_3$ &14, 15, 12, 13, 10, 11, 8, 9, 6, 7, 5, 4, 3, 2, 1, 0
 & 69 \\
\hline
$\mathcal{F}_4$ &14, 15, 12, 13, 10, 11, 9, 8, 6, 7, 5, 4, 3, 2, 0, 1
 & 56 \\
\hline
$\mathcal{F}_5$ &14, 15, 12, 13, 10, 11, 9, 8, 6, 7, 5, 4, 3, 2, 1, 0
 & 48 \\
\hline
$\mathcal{F}_6$ &14, 15, 12, 13, 10, 11, 9, 8, 7, 6, 4, 5, 2, 3, 0, 1
 & 86 \\
\hline
$\mathcal{F}_7$ &14, 15, 12, 13, 10, 11, 9, 8, 7, 6, 4, 5, 2, 3, 1, 0
 & 58 \\
\hline
$\mathcal{F}_8$ &14, 15, 12, 13, 10, 11, 9, 8, 7, 6, 4, 5, 3, 2, 1, 0
 & 46 \\
\hline
$\mathcal{F}_9$ &14, 15, 12, 13, 10, 11, 9, 8, 7, 6, 5, 4, 3, 2, 0, 1
 & 42 \\
\hline
$\mathcal{F}_{10}$ &14, 15, 12, 13, 10, 11, 9, 8, 7, 6, 5, 4, 3, 2, 1, 0
 & 69 \\
\hline
$\mathcal{F}_{11}$ &14, 15, 12, 13, 11, 10, 9, 8, 7, 6, 5, 4, 2, 3, 1, 0
 & 58 \\
\hline
$\mathcal{F}_{12}$ &14, 15, 13, 12, 11, 10, 8, 9, 7, 6, 4, 5, 2, 3, 1, 0
 & 35 \\
\hline
$\mathcal{F}_{13}$ &14, 15, 13, 12, 11, 10, 8, 9, 7, 6, 4, 5, 3, 2, 1, 0
 & 56 \\
\hline
$\mathcal{F}_{14}$ &14, 15, 13, 12, 11, 10, 8, 9, 7, 6, 5, 4, 3, 2, 1, 0
 & 94 \\
\hline
$\mathcal{F}_{15}$ &14, 15, 13, 12, 11, 10, 9, 8, 7, 6, 5, 4, 3, 2, 0, 1
 & 86 \\ 
\hline
$\mathcal{F}_{16}$ &14, 15, 13, 12, 11, 10, 9, 8, 7, 6, 5, 4, 3, 2, 1, 0
  & 206 \\
 \hline
\end{tabular}
\end{center}
\end{scriptsize}
\end{minipage}
\label{fig:listfonction}
}
\end{center}
\caption{Chaotic Functions Candidates with $n=4$}
 \end{figure}

Quality of produced random sequences have been evaluated with the NIST
Statistical  Test Suite  SP
800-22~\cite{RSN+10}.
For all 15 tests of this battery, the significance level $\alpha$ is set to $1\%$:
a \emph{p-value} which is greater than 0.01 is equivalent that the keystream is accepted as random with a confidence of $99\%$.
Synthetic results in Table.~\ref{fig:TEST}
show that  
all these functions successfully
pass this statistical battery of tests.

\begin{table}[th]
\begin{scriptsize}
\begin{tabular}{|*{17}{c|}}
\hline
{Property}& $\mathcal{F}_{1}$ &$\mathcal{F}_{2}$ &$\mathcal{F}_{3}$ &$\mathcal{F}_{4}$ &$\mathcal{F}_{5}$ &$\mathcal{F}_{6}$ &$\mathcal{F}_{7}$ &$\mathcal{F}_{8}$ &$\mathcal{F}_{9}$ &$\mathcal{F}_{10}$ &$\mathcal{F}_{11}$ &$\mathcal{F}_{12}$ &$\mathcal{F}_{13}$ &$\mathcal{F}_{14}$ &$\mathcal{F}_{15}$ &$\mathcal{F}_{16}$ \\
\hline
Frequency &77.9 &15.4 &83.4 &59.6 &16.3 &38.4 &20.2 &29.0 &77.9 &21.3 &65.8 &85.1 &51.4 &35.0 &77.9 &92.4 \\ 
 \hline 
BlockFrequency &88.3 &36.7 &43.7 &81.7 &79.8 &5.9 &19.2 &2.7 &98.8 &1.0 &21.3 &63.7 &1.4 &7.6 &99.1 &33.5 \\ 
 \hline 
CumulativeSums &76.4 &86.6 &8.7 &66.7 &2.2 &52.6 &20.8 &80.4 &9.8 &54.0 &73.6 &80.1 &60.7 &79.7 &76.0 &44.7 \\ 
 \hline 
Runs &5.2 &41.9 &59.6 &89.8 &23.7 &76.0 &77.9 &79.8 &45.6 &59.6 &89.8 &2.4 &96.4 &10.9 &72.0 &11.5 \\ 
 \hline 
LongestRun &21.3 &93.6 &69.9 &23.7 &33.5 &30.4 &41.9 &43.7 &30.4 &17.2 &41.9 &51.4 &59.6 &65.8 &11.5 &61.6 \\ 
 \hline 
Rank &1.0 &41.9 &35.0 &45.6 &51.4 &20.2 &31.9 &83.4 &89.8 &38.4 &61.6 &4.0 &21.3 &69.9 &47.5 &95.6 \\ 
 \hline 
FFT &40.1 &92.4 &97.8 &86.8 &43.7 &38.4 &76.0 &57.5 &36.7 &35.0 &55.4 &57.5 &86.8 &76.0 &31.9 &7.6 \\ 
 \hline 
NonOverlappingTemplate &49.0 &45.7 &50.5 &51.0 &48.8 &51.2 &51.6 &50.9 &50.9 &48.8 &45.5 &47.3 &47.0 &49.2 &48.6 &46.4 \\ 
 \hline 
OverlappingTemplate &27.6 &10.9 &53.4 &61.6 &16.3 &2.7 &59.6 &94.6 &88.3 &55.4 &76.0 &23.7 &47.5 &91.1 &65.8 &81.7 \\ 
 \hline 
Universal &24.9 &35.0 &72.0 &51.4 &20.2 &74.0 &40.1 &23.7 &9.1 &72.0 &4.9 &13.7 &14.5 &1.8 &93.6 &65.8 \\ 
 \hline 
ApproximateEntropy &33.5 &57.5 &65.8 &53.4 &26.2 &98.3 &53.4 &63.7 &38.4 &6.7 &53.4 &19.2 &20.2 &27.6 &67.9 &88.3 \\ 
 \hline 
RandomExcursions &29.8 &35.7 &40.9 &36.3 &54.8 &50.8 &43.5 &46.0 &39.1 &40.8 &29.6 &42.0 &34.8 &33.8 &63.0 &46.3 \\ 
 \hline 
RandomExcursionsVariant &32.2 &40.2 &23.0 &39.6 &47.5 &37.2 &56.9 &54.6 &53.3 &31.5 &23.0 &38.1 &52.3 &57.1 &47.7 &40.8 \\ 
 \hline 
Serial &56.9 &58.5 &70.4 &73.2 &31.3 &45.9 &60.8 &39.9 &57.7 &21.2 &6.4 &15.6 &44.7 &31.4 &71.7 &49.1 \\ 
 \hline 
LinearComplexity &24.9 &23.7 &96.4 &61.6 &83.4 &49.4 &49.4 &18.2 &3.5 &76.0 &24.9 &97.2 &38.4 &38.4 &1.1 &8.6 \\ 
 \hline
\end{tabular}
\end{scriptsize}

\caption{NIST Test Evaluation of PRNG instances}\label{fig:TEST}  
\end{table}

\section{Conclusion and Future Work}
This work has shown that
discrete-time dynamical systems $G_f$ are chaotic iff embedded Boolean
maps $f$ have strongly connected iteration graph $\Gamma(f)$.
Sufficient conditions on its interaction graph $G(f)$ 
have been further proven to ensure this strong connexity.
Finally, we have proven that the output of such a function is uniformly 
distributed iff the induced Markov chain can be represented as 
a double stochastic matrix.
We have applied such a complete theoretical work on chaos to 
pseudo random number generation and all
experiments have confirmed theoretical
results. 
As far as we know, this work is the first one that allows to \emph{compute}
new functions whose chaoticity is proven and preserved during implementation.
The approach relevance has been shown on PRNGs but is not limited to that 
domain.
In fact, this whole work has applications everywhere chaoticity is a 
possible answer,
e.g., in hash functions, digital watermarking\ldots

In a future work, we will investigate whether the characterization of uniform 
distribution may be expressed in terms of interaction graph, avoiding thus to 
generate functions and to check later whether they induce double stochastic 
Markov matrix. 
The impact of the description of chaotic iterations as Markov processes 
will be studied more largely.
We will look for new characterizations concerning other relevant topological properties
of disorder, such as topological entropy, expansivity, Lyapunov exponent, 
instability, etc.
Finally, the relation between these mathematical definitions and intended 
properties for each targeted application will be investigated too, specifically 
in the security field.

\bibliographystyle{splncs}
\bibliography{abbrev,biblio,mabase,biblioand}

\end{document}